\documentclass[preprint,12pt]{elsarticle}

\usepackage{amssymb}
\usepackage{amsthm}

\usepackage{amsbsy}
\usepackage{comment}
\usepackage{algorithmic}
\usepackage[linesnumbered,ruled,vlined]{algorithm2e}
\usepackage{multirow}
\usepackage{rotating}

\newtheorem{theorem}{Theorem}
\newtheorem{lemma}{Lemma}
\newtheorem{definition}{Definition}
\newtheorem{researchQuestion}{RQ}

\newcommand{\ie}{i.e.}

\let\oldnl\nl
\newcommand{\nonl}{\renewcommand{\nl}{\let\nl\oldnl}}

\begin{document}

\begin{frontmatter}



\title{Constrained locating arrays for combinatorial interaction testing}


\author{Hao Jin\corref{cor1}}
\ead{k-kou@ist.osaka-u.ac.jp}
\author{Tatsuhiro Tsuchiya}
\address{1-5 Yamadaoka, Suita-Shi, Osaka 565-0871, Japan}

\begin{abstract}
This paper introduces the notion of \textit{Constrained Locating Arrays (CLAs)},
mathematical objects which can be used for fault localization in software testing.
CLAs extend ordinary locating arrays to make them applicable
to testing of systems that have constraints on test parameters.
Such constraints are common in real-world systems; thus CLA
enhances the applicability of locating arrays to practical
testing problems.
The paper also proposes an algorithm for constructing CLAs.
Experimental results show that the proposed algorithm scales to problems of
practical sizes.
\end{abstract}


\end{frontmatter}


\section{Introduction}
\label{sec:introduction}

\textit{Combinatorial interaction testing} is a well-known strategy for
software testing.
In the strategy, a \textit{System Under Test (SUT)} is modeled as
a finite set of test parameters or \textit{factors}
and  every interaction of interest
is exercised by at least one test. Empirical results suggest
that testing interactions involving a fairly small number of factors,
typically two or three, suffices to reveal most of latent faults. Many
studies have been developed to construct small test sets 
for combinatorial interaction testing.
Such test sets are often called \textit{Covering Arrays (CAs)}.
Surveys on these studies can be found in, for example, \cite{Colbourn2004,Grindal2005,Nie:2011}.

An important direction of extending the capability of combinatorial
interaction testing is to add fault localization capability to it. 
 \textit{Locating Arrays (LAs)} can be used as test suites that provide this capability~\cite{colbourn_locatingarray2008}.
In \cite{colbourn_locatingarray2008} LAs of a few different types are defined.
For example, a $(d, t)$-LA
enables to locate a set of $d$ failure-triggering $t$-way interactions
using the test outcome.

The purpose of this paper is to extend the notion of LAs to expand the
applicability to practical testing problems. Specifically, we propose
\textit{Constrained Locating Arrays (CLAs)} which can be used to detect
and locate failure-triggering interactions in the presence of \textit{constraints}.
Constraints, which prohibit some particular tests, are common in
real-world systems.  Constraint handling has been well studied in the field of combinatorial interaction testing~\cite{8102999}. The main focus
of the previous studies is on constructing test sets, often called
a \textit{Constrained Covering Array (CCA)}, that consist only of
constraint-satisfying tests and cover all interactions that can occur
in constraint-satisfying tests.

CLAs requires additional considerations about constraints. Specifically,
constraints may make it impossible to distinguish a failure-triggering
interaction or set of such interactions
from another; hence a special treatment must be needed to
deal with such an inherently indistinguishable pair.
By extending LAs with the concept of
\textit{distinguishability}, we provide the formal definition
of CLAs.
We also propose a generation method for CLAs and demonstrate that 
the generation method can scale to problems of practical sizes.


The rest of the paper is organized as follows.
Section~\ref{sec:pre} describes the SUT model
and the definition of locating arrays, as well as some related notions.
Section~\ref{sec:cla} presents the definition of
CLAs and some basic theorems about them.
Section~\ref{sec:generation} presents a computational method for generating CLAs. 
Section~\ref{sec:evaluation} shows experimental results obtained by applying the method 
to a number of problem instances.
Section~\ref{sec:related} summarizes related work.
Section~\ref{sec:summary} concludes the paper with possible future directions of work.

\section{Preliminaries}
\label{sec:pre}

\begin{figure}
{\footnotesize
\centering
\begin{tabular}{|l||l|l|l|l|l|}
\hline
factors & $F_1$: Display  & $F_2$: Email  & $F_3$: Camera  & $F_4$: Video  & $F_5$: Video \\
        &                 & \multicolumn{1}{|r|}{Viewer}   &          & \multicolumn{1}{|r|}{Camera}  & \multicolumn{1}{|r|}{Ringtones}  \\ \hline
values  & 0 : 16 MC & 0 : Graphical & 0 : 2 MP & 0 : Yes & 0 : Yes \\
        & 1 : 8 MC  & 1 : Text & 1 : 1 MP & 1 : No  & 1 : No \\
        & 2 : BW    & 2 : None & 2 : None &  &  \\ \hline
constraints
   & \multicolumn{5}{|l|}{$F_2 = 0 \Rightarrow F_1 \neq 2$
     }\\ & \multicolumn{5}{|r|}{
    \hfill {\scriptsize Graphical email viewer requires color display} } \\
   & \multicolumn{5}{|l|}{$F_3 = 0 \Rightarrow F_1 \neq 2$
     }\\ & \multicolumn{5}{|r|}{
     \hfill {\scriptsize 2 Megapixel camera requires color display}} \\
   & \multicolumn{5}{|l|}{$F_2 = 0 \Rightarrow F_3 \neq 0$
     }\\ & \multicolumn{5}{|r|}{
     \hfill {\scriptsize \quad Graphical email viewer not supported with 2 Megapixel camera}} \\
   & \multicolumn{5}{|l|}{$F_1 = 1 \Rightarrow F_3 \neq 0$
     }\\ & \multicolumn{5}{|r|}{
     \hfill {\scriptsize 8 Million color display does not support 2 Megapixel camera}} \\
   & \multicolumn{5}{|l|}{$F_4 = 0 \Rightarrow (F_3 \neq 2 \land F_1 \neq 2)$
     }\\ & \multicolumn{5}{|r|}{
     \hfill {\scriptsize Video camera requires camera and color display}} \\
   & \multicolumn{5}{|l|}{$F_5 = 0 \Rightarrow F_4 = 0$
     }\\ & \multicolumn{5}{|r|}{
     \hfill {\scriptsize Video ringtones cannot occur with No video camera}} \\
   & \multicolumn{5}{|l|}{$\neg(F_1 = 0 \land F_2 =1 \land F_3 = 0)$
     }\\ & \multicolumn{5}{|r|}{
    \hfill  {\scriptsize The combination of 16 Million colors, Text email}} \\
   & \multicolumn{5}{|l|}{
     \hfill {\scriptsize viewer and 2 Megapixel camera will not be supported}}
   \\ \hline
\end{tabular}
}
\caption{Example of an SUT \cite{Cohen:2008}. }
\label{fig:example}
\end{figure}

An SUT is modeled as $\langle \mathcal{F}, \mathcal{S}, \phi\rangle$ where
$\mathcal{F} = \{F_1, F_2, ..., F_k\}$ is a set of factors,
$\mathcal{S} = \{S_1, S_2, ..., S_k\}$ is a set of domains for the factors, and
$\phi: S_1 \times ... \times S_k \rightarrow \{true, false\}$ is a mapping that represents constraints.
Each domain $S_i$ consists of two or more consecutive integers ranging from 0; i.e.,
$S_i = \{0, 1, ..., |S_i| - 1\}$ ($|S_i| > 1$).
A \textit{test} is an element of $S_1 \times S_2 \times ... \times S_k$.
A test $\pmb{\sigma}$ is \textit{valid} if and only if (iff) it satisfies the constraints
$\phi$, i.e., $\phi(\pmb{\sigma}) = true$.
Given an SUT, we denote the set of all valid tests as $\mathcal{R}$.
For a set of $t\ (0\leq t\leq k)$ factors, $\{F_{i_1},...., F_{i_t}\} \subseteq \mathcal{F}$,
the set $\{(i_1, \sigma_1), ...., (i_t, \sigma_t)\}$ such that $\sigma_j \in S_j$ for
all $j$ $(1\leq j \leq t)$ is
a $t$-\textit{way interaction} or an interaction of \textit{strength} $t$.
Hence a test contains or \textit{covers} $({k \atop t})$ $t$-way interactions.
Note that a $k$-way interaction $\{(1, \sigma_1), ..., (k, \sigma_k)\}$ and
a test $\pmb{\sigma} = (\sigma_1, ..., \sigma_k)$ can be treated interchangeably.
Thus we write $T \subseteq \pmb{\sigma}$ iff a test $\pmb{\sigma}$ covers an interaction $T$. 
It should be noted that the only 0-way is the empty set. 
We use $\sqcup$, instead of $\emptyset$, to denote the 0-way interaction. 

Constraints may make it impossible to test some interactions. 
These interactions cannot be covered by any valid tests. 
We call such an interaction \textit{invalid}. 
Formally, an interaction $T$ is \textit{valid} if $T \subseteq \pmb{\sigma}$ 
for some valid test $\pmb{\sigma} \in \mathcal{R}$; it is \textit{invalid}, otherwise. 

As a running example, consider a classic cell-phone example taken from
\cite{Cohen:2008} (Fig.~\ref{fig:example}).
This SUT model has five factors which have three or two values in their domains.
The constraints consist of seven parts.
Test $(1, 0, 1, 1, 1)$, for example, is valid,
whereas test $(1, 0, 0, 0, 1)$ is not valid (invalid)
because it violates the third and fourth constraints.
Similarly, two-way interaction $\{(1, 1), (2, 0)\}$ is valid,
since it occurs in valid test $(1, 0, 1, 1, 1)$.
On the other hand, $\{(2, 0), (3, 0)\}$ is invalid, since
it violates constraint $F_2 = 0 \Rightarrow F_3 \neq 0$ and thus never occurs in any valid tests.

A \textit{test suite} is defined as a (possibly empty) collection of tests and thus can be
represented as an $N \times k$ array $A$ when the number of tests is $N$.
For such an array $A$ and interaction $T$,
we let $\rho_A(T)$ denote the set of tests (rows) of $A$ in which the interaction is covered.
For a set of interactions $\mathcal{T}$,
we define $\rho_A (\mathcal{T}) = \bigcup_{T\in \mathcal{T}}\rho_A(T)$.
We use $\emptyset$ to denote an empty set of interactions. 
Clearly $\rho_A(\emptyset) = \emptyset$. 
(By comparison, $\rho_A(\sqcup)$ is the set of all rows of $A$.)

An interaction is either \textit{failure-triggering} or not.
The result of executing a test $\pmb{\sigma}$ is \textit{fail} iff $\pmb{\sigma}$ covers at least one
failure-triggering interaction;  otherwise the result is \textit{pass}.
Hence the result of executing a test suite $A$ is a vector of size $N$, each element being
either pass or fail.

When there are no constraints, i.e., $\phi(\pmb{\sigma}) = true$ for any
test $\pmb{\sigma} \in S_1 \times \ldots \times S_k$, 
a \textit{Covering Array} (CA) can be used to detect the existence of
fault-triggering interactions of a given strength $t$ or less.
Let $\mathcal{I}_t$ be the set of all $t$-way interactions.
Formally, a $t$-CA is defined by the following condition:

\noindent
\begin{tabular}{p{0.15\textwidth}p{0.85\textwidth}} \\
$t$-CA   & $\forall T \in \mathcal{I}_t$: $\rho_A(T) \not =  \emptyset$
\end{tabular}~\\

On the other hand, a \textit{Locating Array} (LA) can be used to 
locate the set of failure-triggering interactions.
Colbourn and McClary introduced a total of six types of LAs in~\cite{colbourn_locatingarray2008}.
The definitions of the two most basic types of LAs are shown below.

\noindent
\begin{tabular}{lp{0.85\textwidth}} \\
$(d,t)$-LA   & $\forall \mathcal{T}_1, \mathcal{T}_2 \subseteq \mathcal{I}_t$ such that
$|\mathcal{T}_1| = |\mathcal{T}_2| = d$ :
$\rho_A(\mathcal{T}_1) = \rho_A(\mathcal{T}_2) \Leftrightarrow
  \mathcal{T}_1 = \mathcal{T}_2$ \\

$(\overline d,t)$-LA   & $\forall \mathcal{T}_1, \mathcal{T}_2 \subseteq \mathcal{I}_t$ such that
$0 \leq |\mathcal{T}_1| \leq d$, $0 \leq |\mathcal{T}_2| \leq d$ :
$\rho_A(\mathcal{T}_1) = \rho_A(\mathcal{T}_2) \Leftrightarrow
  \mathcal{T}_1 = \mathcal{T}_2$
\end{tabular}~\\

The definition of other two types of LAs, namely $(d,\overline t)$-LAs
and $(\overline d,\overline t)$-LAs, requires the notion of \textit{independence}~\cite{colbourn_locatingarray2008}.
Let
$\overline{\mathcal{I}_t}$ be the set of all interactions of strength at most
$t$, i.e., $\overline{\mathcal{I}_t} = \mathcal{I}_0 \cup \mathcal{I}_1 \cup ... \cup \mathcal{I}_t$.
A set of interactions (interaction set) $\mathcal{T} \subseteq \overline{\mathcal{I}_t}$ is
\textit{independent} iff there do not exist two interactions $T, T' \in \mathcal{T}$ with $T \subset T'$.
For example, consider a set of two interactions
$\{\{(1,1)\}, \{(1,1), (2,0)\}\} \ (\subseteq \overline{\mathcal{I}_2})$ for the running example.
This interaction set is not independent because $\{(1,1)\} \subset \{(1,1), (2,0)\}$.
Note that if two interactions $T, T'$ are both failure-triggering and $T \subset T'$,
then the failure triggered by $T$ always masks the failure triggered by $T'$.
Because of this, it is natural to limit the scope of fault localization to independent interaction sets.
Based on $\overline{\mathcal{I}_t}$ and the notion of independent interaction
sets, the two types of LAs are defined  as follows.

\noindent
\begin{tabular}{lp{0.85\textwidth}} \\
$(d,\overline t)$-LA   & $\forall \mathcal{T}_1, \mathcal{T}_2 \subseteq
\overline{\mathcal{I}_t}$ such that
$|\mathcal{T}_1| = |\mathcal{T}_2| = d$ and $\mathcal{T}_1, \mathcal{T}_2$ are independent:
$\rho_A(\mathcal{T}_1) = \rho_A(\mathcal{T}_2) \Leftrightarrow
  \mathcal{T}_1 = \mathcal{T}_2$ \\

$(\overline d,\overline t)$-LA   & $\forall \mathcal{T}_1, \mathcal{T}_2 \subseteq
\overline{\mathcal{I}_t}$ such that
 $0 \leq |\mathcal{T}_1| \leq d$, $0 \leq |\mathcal{T}_2| \leq d$ and $\mathcal{T}_1, \mathcal{T}_2$ are independent:
 $\rho_A(\mathcal{T}_1) = \rho_A(\mathcal{T}_2) \Leftrightarrow
   \mathcal{T}_1 = \mathcal{T}_2$
\end{tabular}

~\\
We do not consider the remaining two types of locating arrays, namely $(\hat d, t)$-LAs
and $(\hat d, \overline{t})$-LAs, because they either exist in trivial cases
or otherwise are equivalent to $(\overline{d}, t)$- and $(\overline{d}, \overline{t})$-LAs.

Figure~\ref{fig:la} shows a $(1, 2)$-LA for the running example shown in Fig.~\ref{fig:example}.
Let $A$ be the LA and $\pmb{a}_i (1\leq i \leq N = 15)$ be the $i$th row.
If the pass/fail result were obtained for all these tests, any failure-triggering
single two-way interaction could be identified.
For example, if only the first test $\pmb{a}_1$ failed, then the failure-triggering interaction
would be determined to be $\{(2,0), (3,0)\}$, because
$\rho(\mathcal{T}) = \{\pmb{a}_1\}$ holds only for $\mathcal{T} = \{\{(2,0), (3,0)\}\}$,
provided that $|\mathcal{T}| = 1$ and $|T|=2$ for $T\in \mathcal{T}$.
However, this array cannot be used for testing the system because of
the constraints.
For example, $\pmb{a}_1$ is not valid and thus cannot be executed
in reality.

\begin{figure}
\centering
{\footnotesize
\begin{tabular}{ccccc}
\hline
0 & 0 & 0 & 0 & 0 \\
0 & 0 & 1 & 1 & 1 \\
0 & 0 & 2 & 0 & 1 \\
0 & 1 & 1 & 0 & 0 \\
0 & 2 & 0 & 0 & 1 \\
1 & 0 & 1 & 0 & 1 \\
1 & 0 & 2 & 1 & 1  \\
1 & 1 & 0 & 0 & 1 \\
1 & 1 & 2 & 1 & 0  \\
1 & 2 & 0 & 1 & 0 \\
2 & 0 & 2 & 0 & 0  \\
2 & 1 & 1 & 1 & 1 \\
2 & 2 & 0 & 1 & 0 \\
2 & 2 & 1 & 0 & 0 \\
2 & 2 & 2 & 1 & 1 \\
\hline
\end{tabular}
}

\caption{(1,2)-LA for the running example. Constraints are not taken into account.}
\label{fig:la}
\end{figure}



\section{Constrained locating arrays}
\label{sec:cla}

In the presence of constrains, a test suite must consist of only valid tests.
From now on, we assume that an array $A$ representing a test
suite consists of a (possibly empty) set of valid tests.
In practice, this problem has been
circumvented by, instead of CAs, using Constrained Covering Arrays (CCAs).
Let $\mathcal{VI}_t$ be the set of all valid $t$-way interactions.
Then a CCA of strength~$t$, denoted as $t$-CCA, is defined as follows.

\noindent
\begin{tabular}{p{0.15\textwidth}p{0.85\textwidth}} \\
$t$-CCA   & $\forall T \in \mathcal{VI}_t$: $\rho_A(T) \not =  \emptyset$
\end{tabular}~\\~\\
In words, a $t$-CCA is an array that covers all valid interactions of strength~$t$.
It is easy to see that a $t$-CCA, $t\geq 1$ is a $(t-1)$-CCA. Therefore, the 
above definition is equivalent to:

\noindent 
\begin{tabular}{p{0.15\textwidth}p{0.85\textwidth}} \\
	$t$-CCA   & $\forall T \in \overline{\mathcal{VI}_t}$: $\rho_A(T) \not =  \emptyset$
\end{tabular}~\\~\\
Figure~\ref{fig:2cca} shows a 2-CCA for the running example.

When incorporating constraints into LA, it is crucial to take
into consideration, in addition to the presence of invalid interactions,
the fact that constraints may make it impossible to
identify some set of failure-triggering interactions, which
could be identified if no constraints existed.
This requires us the notion of \textit{distinguishability} to formally
define CLAs.

\begin{definition}
  \label{def:1}
A pair of sets of valid interactions, $\mathcal{T}_1$ and $\mathcal{T}_2$,
are \textit{distinguishable}
iff $\rho_A(\mathcal{T}_1) \neq \rho_A(\mathcal{T}_2)$ for some array $A$
consisting of valid tests.
\end{definition}
For the running example, $\mathcal{T}_1 = \{\{(1,0), (3,0)\}\},
\mathcal{T}_2 = \{\{(2,2), (3,0)\}\}$ are not distinguishable
(indistinguishable), since
any valid test contains either both of the two-way interactions or none of them.
That is, tests that cover exactly one of the two interaction sets
(e.g., (0 1 0 0 0) or (1 2 0 0 0)) are all invalid.
Hence no array $A$ exists such that $\rho_A(\mathcal{T}_1) \neq \rho_A(\mathcal{T}_2)$.

It should be noted that even if there are no constraints, there can be some
indistinguishable pairs of interaction sets.
In the running example,
two interaction sets $\{\{(4, 0)\}, \{(4,1)\}\}$, $\{\{(5, 0)\}, \{(5,1)\}\}$
are indistinguishable even if the constraints were removed,
because any test has 0 or 1 on factors $F_4$ and $F_5$.
Another extreme case is when $\mathcal{T}_1$ and $\mathcal{T}_2$ are identical.
Clearly, identical interactions are always indistinguishable.







\begin{definition}
  \label{def:2}
Let $d \geq 0$ and  $0 \leq t \leq k$.
Let $\mathcal{VI}_t$ be the set of all valid $t$-way interactions
and  $\overline \mathcal{VI}_t$ be the set of all valid interactions of strength at most $t$.
An array $A$ that consists of valid tests or no rows is a $(d,t)$-,
$(\overline d,t)$-, $(d, \overline t)$- or
$(\overline d, \overline t)$-CLA iff the corresponding condition shown below holds.

\noindent
\begin{tabular}{lp{0.8\textwidth}}
$(d,t)$-CLA  & $\forall \mathcal{T}_1, \mathcal{T}_2 \subseteq \mathcal{VI}_t$ such that
$|\mathcal{T}_1| = |\mathcal{T}_2| = d$ and $\mathcal{T}_1, \mathcal{T}_2$ are distinguishable:
$\rho_A(\mathcal{T}_1) \neq \rho_A(\mathcal{T}_2)$  \\

$(\overline d,t)$-CLA  & $\forall \mathcal{T}_1, \mathcal{T}_2 \subseteq \mathcal{VI}_t$ such that
$0 \leq |\mathcal{T}_1| \leq d$, $0 \leq |\mathcal{T}_2| \leq d$ and $\mathcal{T}_1, \mathcal{T}_2$ are distinguishable:
$\rho_A(\mathcal{T}_1) \neq \rho_A(\mathcal{T}_2)$  \\

$(d,\overline t)$-CLA  & $\forall \mathcal{T}_1, \mathcal{T}_2 \subseteq
 \overline{\mathcal{VI}}_t$ such that
 $|\mathcal{T}_1| = |\mathcal{T}_2| = d$ and $\mathcal{T}_1, \mathcal{T}_2$ are independent
 and distinguishable:
$\rho_A(\mathcal{T}_1) \neq \rho_A(\mathcal{T}_2)$  \\

 $(\overline d,\overline t)$-CLA  &  $\forall \mathcal{T}_1, \mathcal{T}_2 \subseteq
 \overline{\mathcal{VI}}_t$ such that
 $0 \leq |\mathcal{T}_1| \leq d$, $0 \leq |\mathcal{T}_2| \leq d$ and $\mathcal{T}_1, \mathcal{T}_2$ are independent
 and distinguishable:
$\rho_A(\mathcal{T}_1) \neq \rho_A(\mathcal{T}_2)$  \\

\end{tabular}

\noindent
(In extreme cases where no two such interaction sets $\mathcal{T}_1, \mathcal{T}_2$
exist, any $A$ is a CLA.)
\end{definition}

The intuition of the definition is that
if the SUT has a set of $d$ (or $\leq d$) failure-triggering interactions, then
the test outcome obtained by executing all tests in $A$ will be
different from the one that would be obtained when the SUT had
a different set of $d$ (or $\leq d$) failure-triggering interactions,
unless the two interaction sets are not distinguishable.

The following theorem follows from the definition.
\begin{theorem}\label{theorem:subsumption}
A $(\overline d,\overline t)$-CLA is a $(\overline d, t)$- and $(d,\overline t)$-CLA.
A $(\overline d, t)$-CLA and
a $(d,\overline t)$-CLA are both a $(d, t)$-CLA.
A $(\overline d,\overline t)$-CLA and a $(\overline d, t)$-CLA are
a $(\overline{d -1},\overline t)$-CLA and a $(\overline{d -1}, t)$-CLA, respectively.
A $(\overline d,\overline t)$-CLA
and a $(d, \overline t)$-CLA are
a $(\overline{d},\overline{t-1})$-CLA and a $(d, \overline{t-1})$-CLA, respectively.
\end{theorem}


Theorem~\ref{theorem:equiv} states that
when there are no constraints,
an LA, if existing, and a CLA are equivalent.

\begin{theorem}\label{theorem:equiv}
Suppose that the SUT has no constraints, i.e., $\phi(\pmb{\sigma}) = true$ for
all $\pmb{\sigma} \in V_1 \times \ldots \times V_k$,
and that an LA $A$ exists.
Then 1) $A$ is a CLA with the same parameters,
and 2) any CLA with the same parameters as $A$ is an LA
(which is possibly different from $A$) with the same parameters.
\end{theorem}
\begin{proof}
Suppose that $A$ is a $(d, t)$-LA.
Let $\mathcal{T}_1, \mathcal{T}_2 \subseteq \mathcal{I}_t (= \mathcal{VI}_t)$
be any two interaction sets such that
$|\mathcal{T}_1| = |\mathcal{T}_2| = d$.
1) If $\mathcal{T}_1 \neq \mathcal{T}_2$,
then $\rho_A(\mathcal{T}_1) \neq \rho_A(\mathcal{T}_2)$.
If $\mathcal{T}_1 = \mathcal{T}_2$, then they are not distinguishable.
Hence $A$ is a $(d, t)$-CLA.
2) Suppose that an array $A'$ is a $(d, t)$-CLA.
If $\mathcal{T}_1 \neq \mathcal{T}_2$, then
$\rho_A(\mathcal{T}_1) \neq \rho_A(\mathcal{T}_2)$ and thus they are distinguishable,
which in turn implies $\rho_{A'}(\mathcal{T}_1) \neq \rho_{A'}(\mathcal{T}_2)$.
If $\mathcal{T}_1 = \mathcal{T}_2$, then they are not distinguishable
and trivially  $\rho_{A'}(\mathcal{T}_1) = \rho_{A'}(\mathcal{T}_2)$.
Hence $A'$ is a $(d, t)$-LA.
The same argument applies to the other three types of LAs.
\end{proof}

It should be noted that
a CLA always exists whether there are constraints or not,
as will be shown in Theorem~\ref{theorem:exist}.
On the other hand, LAs do not always exist.
For example, no $(2,1)$-LAs exist for the running example:
Consider $\mathcal{T}_1 = \{\{(4, 0)\}, \{(4, 1)\}\}$ and
$\mathcal{T}_2 = \{\{(5, 0)\}, \{(5, 1)\}\}$.
Then $\rho_A(\mathcal{T}_1)$ and $\rho_A(\mathcal{T}_2)$ both include all rows; thus
$\rho_A(\mathcal{T}_1) = \rho_A(\mathcal{T}_2)$ for any $A$.
The guarantee of the existence of CLAs comes from the definition which exempts
indistinguishable pairs of interaction sets from fault localization.
In that sense, CLAs can be viewed as a ``best effort'' variant of LAs.

\begin{lemma}\label{lemma:dist}
A pair of sets of valid interactions, $\mathcal{T}_1$ and $\mathcal{T}_2$, are distinguishable
iff there is a valid test that covers some interaction in $\mathcal{T}_1$ or $\mathcal{T}_2$
but no interactions in $\mathcal{T}_2$ or $\mathcal{T}_1$, respectively, i.e.,
for some valid test $\pmb{\sigma} \in \mathcal{R}$, $(\exists T \in \mathcal{T}_1: T \subseteq \pmb{\sigma})
        \land (\forall T \in \mathcal{T}_2: T \not \subseteq \pmb{\sigma})$ or
$(\exists T \in \mathcal{T}_2: T \subseteq \pmb{\sigma}) \land
  (\forall T \in \mathcal{T}_1: T \not \subseteq \pmb{\sigma})$.
\end{lemma}
\begin{proof}
(If part) Suppose that there is such a valid test $\pmb{\sigma}$. Consider an array
$A$ that contains $\pmb{\sigma}$.
Then,  either $\pmb{\sigma} \in \rho_A(\mathcal{T}_1) \land \pmb{\sigma} \not\in \rho_A(\mathcal{T}_2)$
or  $\pmb{\sigma} \not\in \rho_A(\mathcal{T}_1) \land \pmb{\sigma} \in \rho_A(\mathcal{T}_2)$; thus
$\rho_A(\mathcal{T}_1) \neq \rho_A(\mathcal{T}_2)$.
(Only if part) Suppose that there is no such valid test, i.e.,
for every valid test $\pmb{\sigma}$,
$(\forall T \in \mathcal{T}_1: T \not\subseteq \pmb{\sigma})
        \lor (\exists T \in \mathcal{T}_2: T \subseteq \pmb{\sigma})$
and
$(\forall T \in \mathcal{T}_2: T \not\subseteq \pmb{\sigma}) \lor (\exists T \in \mathcal{T}_1: T \subseteq \pmb{\sigma})$.
This means that for every valid test $\pmb{\sigma}$,
$(\forall T \in \mathcal{T}_1: T \not\subseteq \pmb{\sigma}) \land (\forall T \in \mathcal{T}_2: T \not\subseteq \pmb{\sigma})$
or
$(\exists T \in \mathcal{T}_1: T \subseteq \pmb{\sigma}) \lor (\exists T \in \mathcal{T}_2: T \subseteq \pmb{\sigma})$.
Hence for any test $\pmb{\sigma}$ in $A$, $\pmb{\sigma} \not \in \rho_A(\mathcal{T}_1) \land \pmb{\sigma} \not \in \rho_A(\mathcal{T}_2)$
or $\pmb{\sigma} \in \rho_A(\mathcal{T}_1)
\land \pmb{\sigma} \in \rho_A(\mathcal{T}_2)$.
As a result, for any $A$, $\rho_A(\mathcal{T}_1) = \rho_A(\mathcal{T}_2)$.
\end{proof}

\begin{theorem}\label{theorem:exist}
If $A$ is an array consisting of all valid tests, then $A$ is a CLA
with any parameters.
\end{theorem}
\begin{proof}
Let $\mathcal{T}_1$ and $\mathcal{T}_2$ be any interaction sets that are distinguishable.
By Lemma~\ref{lemma:dist}, a valid test $\pmb{\sigma}$ exists such that
$(\exists T \in \mathcal{T}_1: T \subseteq \pmb{\sigma})
        \land (\forall T \in \mathcal{T}_2: T \not \subseteq \pmb{\sigma})$
or
$(\exists T \in \mathcal{T}_2: T \subseteq \pmb{\sigma}) \land
(\forall T \in \mathcal{T}_1: T \not \subseteq \pmb{\sigma})$.
Since $A$ contains this test and by the same argument of the proof of the if-part of Lemma~\ref{lemma:dist},
$\rho_A(\mathcal{T}_1) \neq \rho_A(\mathcal{T}_2)$.
\end{proof}
Although Theorem~\ref{theorem:exist} guarantees that a test suite
consisting of all valid tests is a CLA, it is desirable to
use a smaller test suite in practice.
In Section~\ref{sec:generation}, we present a computational method
for generating small CLAs.





\subsection{Examples of CLAs} 

Here we show (1,1)-, $(\overline{2}, 1)$-, $(1, \overline{2})$- and $(\overline{2}, \overline{2})$-CLAs for the running SUT example.
Figures~\ref{fig:cla11}, \ref{fig:cla21}, \ref{fig:cla12} and \ref{fig:cla22} respectively show these CLAs.
The sizes (i.e., the number of rows) of these arrays are 5, 12, 15 and 28.
The number of valid tests for the running example is~31; thus
these CLAs, except the $(\overline{2}, \overline{2})$-CLA, are
considerably smaller than the array that consists of all valid tests.
On the other hand, the $(\overline{2}, \overline{2})$-CLA is almost as large as
the exhaustive one.
The three missing valid tests are (0, 2, 1, 0, 1), (0, 2, 1, 1, 1) and (1, 1, 1, 1, 1).

One can verify that these are indeed CLAs by checking the necessary and sufficient
conditions using the facts shown below. For the running example,
all interactions of strength $\leq 2$ are valid,
except ten two-way interactions
listed below.
\begin{quote}
	\begin{tabular}{cccc}
		$\{(1, 2), (2, 0)\}$
		&			$\{(1, 1), (3, 0)\}$
		&			$\{(1, 2), (3, 0)\}$
		&			$\{(1, 2), (4, 0)\}$ \\
		$\{(1, 2), (5, 0)\}$
		&			$\{(2, 0), (3, 0)\}$
		&			$\{(2, 1), (3, 0)\}$
		&			$\{(3, 2), (4, 0)\}$ \\
		$\{(3, 2), (5, 0)\}$
		&			$\{(4, 1), (5, 0)\}$
	\end{tabular}
\end{quote}
For the example,
all pairs $\mathcal{T}_1, \mathcal{T}_2 \subseteq {\mathcal{VI}}_{t=1}$
such that $\mathcal{T}_1 \neq \mathcal{T}_2$ and $|\mathcal{T}_1| = |\mathcal{T}_2| = d = 1$ are
distinguishable.
That is, any pair of distinct one-way interactions are distinguishable.
Figure~\ref{fig:pairs} shows pairs of interaction sets that are \textit{not} distinguishable
for the other parameters $d (\overline{d}), t (\overline{t})$.

\begin{figure}
	\centering
	{\footnotesize
		\begin{tabular}{ccccc}
			\hline
			0  &    0  &    1  &    0  &    1\\
			0  &    0  &    2  &    1  &    1\\
			0  &    1  &    1  &    0  &    0\\
			0  &    2  &    0  &    0  &    0\\
			0  &    2  &    0  &    1  &    1\\
      1  &    0  &    1  &    0  &    0\\
      1  &    1  &    1  &    0  &    0\\
      1  &    1  &    2  &    1  &    1\\
      1  &    2  &    2  &    1  &    1\\
      2  &    1  &    2  &    1  &    1\\
      2  &    2  &    1  &    1  &    1\\
			\hline
		\end{tabular}
	}
	\caption{2-CCA for the running example.}
	\label{fig:2cca}
\end{figure}

\begin{figure}
	\centering
	{\footnotesize
		\begin{tabular}{ccccc}
			\hline
			0  &    0  &    1  &    1  &    1\\
			0  &    2  &    0  &    0  &    1\\
			1  &    1  &    1  &    0  &    0\\
			1  &    2  &    2  &    1  &    1\\
			2  &    1  &    1  &    1  &    1\\
			\hline
		\end{tabular}
	}
	\caption{(1,1)-CLA for the running example.}
	\label{fig:cla11}
\end{figure}

\begin{figure}
	\centering
	{\footnotesize
		\begin{tabular}{ccccc}
			\hline
			0  &    0  &    1  &    0  &    0\\
			0  &    0  &    2  &    1  &    1\\
			0  &    1  &    1  &    0  &    0\\
			0  &    1  &    2  &    1  &    1\\
			0  &    2  &    0  &    0  &    0\\
			0  &    2  &    0  &    0  &    1\\
			0  &    2  &    0  &    1  &    1\\
			1  &    0  &    1  &    0  &    1\\
			1  &    2  &    1  &    0  &    0\\
			1  &    2  &    2  &    1  &    1\\
			2  &    1  &    1  &    1  &    1\\
			2  &    2  &    2  &    1  &    1\\
			\hline
		\end{tabular}
	}
	\caption{$(\overline{2},1)$-CLA for the running example.}
	\label{fig:cla21}
\end{figure}

\begin{figure}
	\centering
	{\footnotesize
		\begin{tabular}{ccccc}
			\hline
			0  &    0  &    1  &    0  &    0\\
			0  &    0  &    2  &    1  &    1\\
			0  &    1  &    1  &    0  &    1\\
			0  &    2  &    0  &    0  &    0\\
			0  &    2  &    0  &    0  &    1\\
			0  &    2  &    0  &    1  &    1\\
			0  &    2  &    2  &    1  &    1\\
			1  &    0  &    1  &    0  &    1\\
			1  &    0  &    1  &    1  &    1\\
			1  &    1  &    1  &    0  &    0\\
			1  &    1  &    2  &    1  &    1\\
			1  &    2  &    1  &    0  &    0\\
			2  &    1  &    1  &    1  &    1\\
			2  &    1  &    2  &    1  &    1\\
			2  &    2  &    1  &    1  &    1\\
			\hline
		\end{tabular}
	}

	\caption{$(1,\overline{2})$-CLA for the running example.}
	\label{fig:cla12}
\end{figure}

\begin{figure}
	\centering
	\begin{tabular}{ccccc}
		\hline
		0 & 0 & 1 & 0 & 0\tabularnewline
		0 & 0 & 1 & 0 & 1\tabularnewline
		0 & 0 & 1 & 1 & 1\tabularnewline
		0 & 0 & 2 & 1 & 1\tabularnewline
		0 & 1 & 1 & 0 & 0\tabularnewline
		0 & 1 & 1 & 0 & 1\tabularnewline
		0 & 1 & 1 & 1 & 1\tabularnewline
		0 & 1 & 2 & 1 & 1\tabularnewline
		0 & 2 & 0 & 0 & 0\tabularnewline
		0 & 2 & 0 & 0 & 1\tabularnewline
		0 & 2 & 0 & 1 & 1\tabularnewline
		0 & 2 & 1 & 0 & 0\tabularnewline
		0 & 2 & 2 & 1 & 1\tabularnewline
		1 & 0 & 1 & 0 & 0\tabularnewline
		1 & 0 & 1 & 0 & 1\tabularnewline
		1 & 0 & 1 & 1 & 1\tabularnewline
		1 & 0 & 2 & 1 & 1\tabularnewline
		1 & 1 & 1 & 0 & 0\tabularnewline
		1 & 1 & 1 & 0 & 1\tabularnewline
		1 & 1 & 2 & 1 & 1\tabularnewline
		1 & 2 & 1 & 0 & 0\tabularnewline
		1 & 2 & 1 & 0 & 1\tabularnewline
		1 & 2 & 1 & 1 & 1\tabularnewline
		1 & 2 & 2 & 1 & 1\tabularnewline
		2 & 1 & 1 & 1 & 1\tabularnewline
		2 & 1 & 2 & 1 & 1\tabularnewline
		2 & 2 & 1 & 1 & 1\tabularnewline
		2 & 2 & 2 & 1 & 1\tabularnewline
		\hline
	\end{tabular}
	\caption{$(\overline{2}, \overline{2})$-CLA for the running example.}
	\label{fig:cla22}
\end{figure}

\begin{figure}
	\centering

	\begin{tabular}{cc}
		\{\{(1, 0)\}\},\
		\{\{(1, 0)\}, \{(3, 0)\}\} &

		\{\{(2, 2)\}\},\
		\{\{(2, 2)\}, \{(3, 0)\}\} \\

		\{\{(4, 0)\}\},\
		\{\{(4, 0)\}, \{(5, 0)\}\} &

		\{\{(4, 1)\}\},\
		\{\{(1, 2)\}, \{(4, 1)\}\} \\

		\{\{(4, 1)\}\},\
		\{\{(3, 2)\}, \{(4, 1)\}\} &

		\{\{(5, 1)\}\},\
		\{\{(1, 2)\}, \{(5, 1)\}\} \\

		\{\{(5, 1)\}\},\
		\{\{(3, 2)\}, \{(5, 1)\}\} &

		\{\{(5, 1)\}\},\
		\{\{(4, 1)\}, \{(5, 1)\}\} \\

		\{\{(1, 2)\}, \{(4, 1)\}\},\
		\{\{(3, 2)\}, \{(4, 1)\}\} &

		\{\{(1, 2)\}, \{(5, 1)\}\},\
		\{\{(3, 2)\}, \{(5, 1)\}\} \\

		\{\{(1, 2)\}, \{(5, 1)\}\},\
		\{\{(4, 1)\}, \{(5, 1)\}\} &

		\{\{(3, 2)\}, \{(5, 1)\}\},\
		\{\{(4, 1)\}, \{(5, 1)\}\} \\

		\{\{(4, 0)\}, \{(4, 1)\}\},\
		\{\{(4, 0)\}, \{(5, 1)\}\} &

		\{\{(4, 0)\}, \{(4, 1)\}\},\
		\{\{(5, 0)\}, \{(5, 1)\}\} \\

		\{\{(4, 0)\}, \{(5, 1)\}\},\
		\{\{(5, 0)\}, \{(5, 1)\}\} &
	\end{tabular}

	~\\
	(a) $\mathcal{T}_1, \mathcal{T}_2 \subseteq {\mathcal{VI}}_{t=1}$
	such that $|\mathcal{T}_1| = |\mathcal{T}_2| \leq d = 2$ and $\mathcal{T}_1$ and $\mathcal{T}_2$
	are distinct but indistinguishable.

	~\\~\\

	\begin{tabular}{cc}
		\{\{(1, 0), (3, 0)\}\},\  \{\{(2, 2), (3, 0)\}\} &
		\{\{(1, 2), (4, 1)\}\},\  \{\{(1, 2), (5, 1)\}\} \\
		\{\{(3, 2), (4, 1)\}\},\  \{\{(3, 2), (5, 1)\}\} \\
	\end{tabular}

	~\\
	(b) $\mathcal{T}_1, \mathcal{T}_2 \subseteq \overline{\mathcal{VI}}_{t=2}$
	such that $|\mathcal{T}_1| = |\mathcal{T}_2| = d = 1$ and $\mathcal{T}_1$ and $\mathcal{T}_2$
	are distinct but indistinguishable. (Since $|\mathcal{T}_1| = |\mathcal{T}_2| = 1$,
	$\mathcal{T}_1, \mathcal{T}_2$ are trivially independent.)

	~\\~\\

	\begin{tabular}{l}

		\{ \{(1,0)\} \},\  \{ \{(1,0)\}, \{(3,0)\} \} \\

		\{ \{(1,2), (4,1)\} \},\  \{ \{(1,2), (4,1)\}, \{(1,2), (5,1)\}  \} \\

		\{ \{(1,0), (2,0)\}, \{(1,1), (2,0)\} \},\  \{ \{(1,2), (4,1)\}, \{(1,2), (5,1)\}  \} \\
		...

	\end{tabular}

	~\\
	(c) Some examples of $\mathcal{T}_1, \mathcal{T}_2 \subseteq \overline{\mathcal{VI}}_{t=2}$
	such that $|\mathcal{T}_1| = |\mathcal{T}_2| \leq d = 2$ and $\mathcal{T}_1$ and $\mathcal{T}_2$
	are independent and indistinguishable.

	\caption{Indistinguishable pairs of sets of interactions }
	\label{fig:pairs}
\end{figure}

\section{Computational generation of $(\overline{1}, t)$-CLAs}
\label{sec:generation}


In this and next sections, we focus our attention on generation of $(\overline{1}, \overline{t})$-CLAs 
for practical reasons as follows.
As demonstrated in the previous section, when the value of $d$ (or $\overline{d}$)  exceeds 
one, the size of CLAs may become substantially larger than $t$-CCAs, offsetting the very benefit of 
combinatorial interaction testing. 
Also, practical test suites must distinguish the situation where no fault exists 
from that where some hypothesized fault occurs; thus we consider $(\overline{1}, \overline{t})$-CLAs, 
instead of $(1, \overline{t})$-CLAs. 

In this section, we propose an algorithm for generating $(\overline{1}, \overline{t})$-CLAs. 
Although not much research exists on generation of LAs,
there has already been a large body of research on CCA generation in the combinatorial interaction testing field. 
The idea of the proposed algorithm is to make use of an existing CCA generation algorithm
to generate $(\overline{1}, \overline{t})$-CLAs. 
This becomes possible by the theorem below, which establishes the relations 
between CCAs and  $(\overline{1}, \overline{t})$-CLAs. 

\begin{theorem}\label{theorem:cca2cla}
	Let $t$ be an integer such that $0 \leq t < k$.
	If an $N \times k$ array $A$ is a $(t+1)$-CCA, then $A$ is also a $(\overline{1}, \overline{t})$-CLA.
\end{theorem}

\begin{proof}
	Recall that an array $A$ is a $(\overline{1}, \overline{t})$-CLA iff
	$\rho_A(\mathcal{T}_1) \neq \rho_A(\mathcal{T}_2)$ for all 
	$\mathcal{T}_1, \mathcal{T}_2 \subseteq \overline{\mathcal{VI}_t}$ such that
	$0 \leq |\mathcal{T}_1| \leq 1$, $0 \leq |\mathcal{T}_2| \leq 1$, and $\mathcal{T}_1$ and $\mathcal{T}_2$ are distinguishable. (See Definition~\ref{def:2}. Note that $\mathcal{T}_1$ and $\mathcal{T}_2$ are independent, 
	since they contain at most one interaction.)

	Now suppose that an $N\times k$ array $A$ is a $(t+1)$-CCA  such that $0 \leq t < k$.
	If $|\mathcal{T}_1|= |\mathcal{T}_2| = 0$, then $\mathcal{T}_1 = \mathcal{T}_2 = \emptyset$
	and thus they are not distinguishable.
	If $|\mathcal{T}_1|= 1$ and $|\mathcal{T}_2| = 0$, then $\rho_A(\mathcal{T}_1) \neq \emptyset$ 
	because $A$ is a $(t+1)$-CCA and thus any $T \in \overline{\mathcal{VI}_{t+1}}$ is covered by some row in $A$. 
	Since $\rho_A(\emptyset) = \emptyset$, $\rho_A({\cal T}_1) \neq \rho_A({\cal T}_2) = \emptyset$ holds for any ${\cal T}_1$, ${\cal T}_2 \subseteq \overline{{\cal VT}_t}$ if $|\mathcal{T}_1|= 1$ and $|\mathcal{T}_2| = 0$.
	The same argument clearly holds if $|\mathcal{T}_1|= 0$ and $|\mathcal{T}_2| = 1$.

	In the rest of the proof, we consider the case in which  $|\mathcal{T}_1|= 1$ and $|\mathcal{T}_2| = 1$. 
	We will show that $\rho_A(T_a)\not=\rho_A(T_b)$ (i.e. $\rho_A(\{T_a\})\not=\rho_A(\{T_b\})$) always holds for any $T_a, T_b \in \overline{{\cal VI}_t}$ if $\{T_a\}$ and $\{T_b\}$ are distinguishable.
	Let $T_a=\{(F_{a_1},u_{a_1}),\dots,(F_{a_l},u_{a_l})\}$ and $T_b=\{(F_{b_1},v_{b_1}),\dots,(F_{b_m},v_{b_m})\}$
	($0\leq l, m \leq t$).
	Also let ${\rm F}=\{F_{a_1},\dots,F_{a_l}\}\cap\{F_{b_1},\dots,F_{b_m}\}$; {\it i.e.},
	${\rm F}$ is the set of factors that are involved in both interactions.
	There are two cases to consider.
	
	\begin{itemize}
		\item[(1)] For some $F_i\in{\rm F}, u_i\not=v_i$. That is, the two interactions have different values 
		on some factor $F_i$. 
		In this case, $T_a$ and $T_b$ never occur in the same test.
		Since $A$ is a $(t+1)$-CCA, $\rho_A(T_a) \neq \emptyset$ and $\rho_A(T_b) \neq \emptyset$.
		Hence, $\rho_A(T_a)\not=\rho_A(T_b)$.
		
		\item[(2)] ${\rm F}=\emptyset$ or for all $F_i\in{\rm F}, u_i=v_i$. 
		That is, the two interactions have no common factors or have the same value for every factor in common. 
		Since $\{T_a\}$ and $\{T_b\}$ are distinguishable,
		there must be at least one valid test $\boldsymbol\sigma$ in $\mathcal{R}$ that covers 
		either $T_a$ or $T_b$ but not both. 
		Suppose that $\boldsymbol{\sigma}$ covers $T_a$ but does not cover $T_b$. 
		In this case, there is a factor $F_j\in\{F_{b_1},\dots,F_{b_m}\}\backslash{\rm F}$ such that the value on $F_j$ of $\boldsymbol\sigma$, denoted $w_j$,
		is different from $v_j$, because otherwise $T_b$ were covered by $\boldsymbol\sigma$.
		Now consider a $(l+1)$-way interaction $T_a'=T_a\cup\{(F_j,w_j)\}$.
		Since the valid test $\boldsymbol\sigma$ covers $T_a'$, $T_a'$ is a $(l+1)$-way valid interaction.
		Since $A$ is a $(t+1)$-CCA and $l + 1 \leq t + 1$, 
		$A$ contains at least one row that covers $T_a'$.
		This row covers $T_a$ but does not cover $T_b$ because the value on $F_j$ is $w_j$ and $w_j\not=v_j$.
		Hence, $\rho_A(T_a)\not=\rho_A(T_b)$. 
		The same argument applies to the case in which $\boldsymbol{\sigma}$ covers $T_b$ but not $T_a$
	\end{itemize}
	As a result, $\rho_A({\cal T}_1) \neq \rho_A({\cal T}_2)$ holds for any ${\cal T}_1$, ${\cal T}_2 \subseteq \overline{{\cal VT}_t}$ if $|\mathcal{T}_1|= |\mathcal{T}_2| = 1$ and they are distinguishable.
\end{proof}

This theorem shows that a $(t+1)$-CCA is also a $(\overline{1},\overline{t})$-CLA, which 
means that one could use existing CCA generation algorithms to obtain $(\overline{1},\overline{t})$-CLAs; 
but a better approach is possible than simply using a $(t+1)$-CCA as a $(\overline{1},t)$-CLA, because
$(t+1)$-CCAs usually have tests that are redundant in locating failure-triggering interactions of strength $t$ or less.
Specifically, we propose a two-step approach as follows:
First, a $(t+1)$-CCA is generated using an off-the-shelf algorithm.
Then the generated CCA is optimized by removing redundant tests.

The following theorem is useful to check whether a test is redundant or not. 
\begin{theorem}\label{theorem:t2bart}
	If an $N \times k$ array $A$ is a $(\overline{1}, t)$-CLA such that $1 \leq t \leq k$, then 
	$A$ is a $(\overline{1}, \overline{t})$-CLA. 
\end{theorem}
\begin{proof}
	See the appendix. 
\end{proof}
The theorem claims that a $(\overline{1}, t)$-CLA and a $(\overline{1}, \overline{t})$-CLA are equivalent.
This property is useful for simplifying the check. 
A test is determined to be redundant if its removal does not invalidate 
the condition required for the array to be a $(\overline{1}, \overline{t})$-CLA. 
Because of the equivalence of $(\overline{1}, t)$-CLAs and $(\overline{1}, \overline{t})$-CLAs, 
we can restrict the interactions to be considered to those in $\mathcal{VI}_t$, instead of $\overline{\mathcal{VI}_t}$.

\begin{algorithm}[tb]
	\caption{Algorithm for CLA generation \label{alg:main}}
	\DontPrintSemicolon
	\KwIn{SUT ${\cal M}$, strength $t$}
	\KwOut{\mbox{($\overline{1}, \overline{t}$)}-CLA $A$}
	$A \gets $ \textsc{generateCCA}$({\cal M}, t+1)$\label{alg:ccagen}\;
	\nonl \quad // generate (t+1)-CCA \;

	${\cal VI}_t \gets$  \textsc{getAllInteractions}$(A, t)$\;
	\nonl \quad // get all t-way interactions from the (t+1)-CCA\;
	
	$map \gets \textsc{mapInteractionToRows}({\cal VI}_t, A)$\;
	\nonl \quad // get a mapping that maps $T \in {\cal VI}_t$ to a set of rows $\rho_A(T)$ \;

	\For{{\rm each} row $\boldsymbol{\sigma} \in A$}{
			\nonl// randomly pick a row that has yet to be selected \;
		$map' \gets \textsc{updateMap}(map,  \boldsymbol\sigma)$\;
				\quad	\nonl// get a mapping for the array with $\boldsymbol{\sigma}$ removed \;
		$\mathcal{I}$ $\gets$ \textsc{getInteractions}($\boldsymbol\sigma$, $t$)\;
				\quad	\nonl// get all $t$-way interactions that appear in $\boldsymbol{\sigma}$  \;
		\If{$(\forall T\in{\cal I}: map'(T)\not=\emptyset) \land$ \;
		\quad	 $(\forall T_a \in{\cal I}, \forall T_b \in{\cal VI}_t:
		map(T_a)\not=map(T_b) \Rightarrow 
		map'(T_a)\not=map'(T_b))$}{
			\nonl// test $\boldsymbol\sigma$ is redundant \;
					$A \gets A$ with $\boldsymbol\sigma$ removed \;
			$map \gets map'$
		}
	}
	\Return{A}
\end{algorithm}

Algorithm \ref{alg:main} generates a $(\overline{1},\overline{t})$-CLA using this approach.
The algorithm takes an SUT model $\mathcal{M}$ and strength $t$ as input and finally returns a $(\overline{1},\overline{t})$-CLA $A$.

In the first line of the algorithm,
the function \textsc{generateCCA()}  uses an existing algorithm to generate a $(t+1)$-CCA.
Then the function \textsc{getAllInteractions()}  is called to enumerate all $t$-way interactions the $(t+1)$-CCA contains.
The interactions obtained are the set of all valid $t$-way interactions (i.e., $\mathcal{VI}_t$), because all interactions occurring in a CCA are valid and any $(t+1)$-CCA contains all $t$-way valid interactions.
Once all the valid $t$-way interactions have been collected,
we compute a mapping $map$ that maps each of them to the set of rows of $A$ that cover it; 
that is, $map: T \mapsto \rho_A(T)$ where $T \in \mathcal{VI}_t$.


In each iteration of the for loop, a row $\boldsymbol\sigma$ is randomly chosen from $A$.
Then we compute $map'$ which is a mapping such that $map': T \mapsto \rho_A(T) \backslash \{\boldsymbol{\sigma}\}$.
In other words, $map'$ is $\rho_{A'}(T)$ where $A'$ is the array  
obtained from $A$ by removing $\boldsymbol{\sigma}$
from it.
The function \textsc{updateMap()} is used to obtain $map'$. 
Also we enumerate all $t$-way interactions that are covered by $\boldsymbol\sigma$. 
The set of these interactions is represented by $\mathcal{I}$.

In each iteration of the loop, we check whether $\boldsymbol{\sigma}$ can be removed or not. 
The row can be removed if $A$ remains to be a $(\overline{1}, t)$-CLA 
(equivalently, $(\overline{1}, \overline{t})$-CLA) after the removal. 
This check is performed by checking two conditions. 

One condition is that every valid $t$-way interaction $T$ still has some row that covers it; 
i.e., $map'(T)\not=\emptyset$.
The condition holds if and only if $\rho_{A'}(\mathcal{T}_1) \neq \rho_{A'}(\mathcal{T}_2)$ holds when 
$|\mathcal{T}_1| = 0$ and $|\mathcal{T}_2| = 1$, since 
$|\mathcal{T}_1| = 0$ implies $\mathcal{T}_1 = \emptyset$ which in turn implies $\rho_{A'}(\mathcal{T}_1)  = \emptyset$.

The other condition corresponds to the case $|\mathcal{T}_1| = |\mathcal{T}_2| =1$: 
The condition is that for every pair of valid, mutually distinguishable $t$-way interactions, 
they still have different sets of rows in which they are covered. 
In other words, for  $T_a, T_b \in \mathcal{VI}_t$, 
if $\{T_a\}$ and $\{T_b\}$ are distinguishable, 
then $map'(T_a) \neq map'(T_b)$
 (i.e., $\rho_{A'}(T_a) \neq \rho_{A'}(T_b)$). 
Note that 
$\{T_a\}$ and $\{T_b\}$ are distinguishable iff $map(T_a) \neq map(T_b)$
(i.e., $\rho_{A}(T_a) \neq \rho_{A}(T_b)$), since $A$ is a $(\overline{1}, \overline{t})$-CLA.

Clearly, if an interaction $T$ is not covered by $\boldsymbol{\sigma}$, 
the deletion of $\boldsymbol\sigma$ does not alter the set of rows that cover $T$. 
Hence checking of the first condition can be performed by examining only the interactions covered
by $\boldsymbol{\sigma}$, i.e., those in $\mathcal{I}$, instead of all interactions in $\mathcal{VI}_t$. 
The same is true for checking of the second condition: 
it can be performed by checking each pair of an interaction $T_a$ in $\mathcal{I}$ and another interaction 
$T_b \in \mathcal{VI}_t$.

 
The loop is iterated  until all rows in the initial $A$ have been examined. 
Finally, the resulting $A$ becomes a $(\overline{1}, \overline{t})$-CLA of reduced size.

It should be noted that output $(\overline{1}, \overline{t})$-CLAs vary for different runs 
of the algorithm, even if the initial $A$ (i.e., the $(t+1)$-CCA generated in line~1) is identical for all runs.  
This is because the $(\overline{1}, \overline{t})$-CLAs finally yielded depends 
also on the order of deleting rows. 
For example, suppose that there are only three valid $t$-way interactions $T_1, T_2$ and $T_3$ 
and that $\{T_1\}$, $\{T_2\}$, $\{T_3\}$ are distinguishable with each other. 
Also suppose that after mapping each interaction to rows,
we have  $map(T_1)=\{1,2,3\}, map(T_2)=\{1,2,4\}$ and $map(T_3)=\{4,5\}$.
If the order of deleting rows is $1\rightarrow 2\rightarrow 3\rightarrow 4\rightarrow 5$,
rows 1 and 2 are deleted but rows 3, 4 and 5 are not. 
This is because after deleting rows 1 and 2,
the mapping becomes: $map(T_1)=\{3\}, map(T_2)=\{4\}$ and $map(T_3)=\{4,5\}$; 
thus any further deletion of rows would make some interaction lose all its covering rows 
or make identical the sets of covering rows for some pair of interactions.
However, if the deleting order is $5\rightarrow 4\rightarrow 3\rightarrow 2\rightarrow 1$,
rows 5, 3 and 2 are deleted. The deleting order of rows thus influences the sizes of resulting CLAs.
\section{Evaluation}
\label{sec:evaluation}

In this section, the proposed generation algorithm is evaluated.
Here, we focus on the case $t=2$, \ie,
the generation of $(\overline{1} ,\overline{2})$-CLAs.
The evaluation is performed with respect to two criteria: generation time and sizes (the number of rows) of CLAs.
For comparison, we choose the generation algorithm based on an SMT (Satisfiability Modulo Theories) solver which we have proposed in \cite{Jin2018},
because, to our knowledge, there does not exist another method that generates CLAs.

\subsection{SMT-based generation algorithm}
\label{smt-based}

The SMT-based generation algorithm can be regarded as an adaptation of
constraint solving-based methods for generating CCAs~\cite{nanba_satsolving2012,Banbara:2010} or LAs~\cite{Konishi2017,DBLP:journals/corr/abs-1904-07480}.
In this algorithm, the necessary and sufficient conditions for the existence of a $(\overline{1},t)$-CLA 
(which is equivalent to a $(\overline{1},\overline{t})$-CLA) of a given size $N$
are encoded into a conjunction of logic expressions.
Then, the algorithm uses an SMT solver to find a satisfiable valuation of variables of the logic expressions.
If a satisfiable valuation is found, then it can be interpreted as a CLA.
On the other hand, if there is no satisfiable valuation, then
the non-existence of a CLA of size~$N$ can be concluded.


In the encoding of the conditions of a CLA, each cell of the array is represented as a variable; thus
the array is encoded as a set of $N\times k$ variables.
According to the definition of $(\overline{1}, t)$-CLAs (see Definition \ref{def:2}),
three sets of logic expressions are needed.
One of the three sets enforces that all rows of the array satisfy all of the SUT constraints.
Another one is used to guarantee that each valid $t$-way interaction is covered by at least one row.
This ensures that for every $T \in \mathcal{VI}_t$, $\rho_A(\{T\}) \neq \emptyset = \rho_A(\emptyset)$.
The last one enforces that for every pair of valid $t$-way interactions, $T_a, T_b \in \mathcal{VI}_t$, if
$\{T_a\}$ and $\{T_b\}$ are mutually distinguishable,
there is at least one row covering only one interaction of the pair, i.e., $\rho_A(T_a) \neq \rho_A(T_b)$.

In our experiments, if a $(\overline{1} ,\overline{2})$-CLA is successfully generated within a timeout period,
we will decrease $N$ by 1 and repeats runs of the algorithm until
the SMT solver proves the non-existence of CLAs of size~$N$. 
If a run of the algorithm fails to terminate within the timeout period, the repetition is stopped. 



\subsection{Research Questions and Experiment Settings}

We pose several research questions as follows for better understanding of experimental results.

\begin{researchQuestion}
  \label{rq1}
  How does the proposed algorithm perform with respect to generation time and sizes for generated CLAs?
\end{researchQuestion}

\begin{researchQuestion}
  \label{rq2}
  How different is the performance between the proposed algorithm and the SMT-based algorithm?
\end{researchQuestion}

\begin{researchQuestion}
  \label{rq3}
  Does the proposed algorithm scale to real-world problems?
\end{researchQuestion}


We performed experiments where we applied both algorithms
to a total of 30 problem instances, numbered from 1 to 30.
Benchmarks No.1\--5 are provided as part of the CitLab tool~\cite{Gargantini:citlab:iwct2012}.
Benchmarks No.6\--25 can be found in \cite{Segall2011}.
Large benchmarks, namely, benchmarks No.26\--30 are taken from \cite{Cohen:2008}.
For each problem instance the proposed algorithm was executed 10 times, as it is a nondeterministic algorithm.
On the other hand, the SMT-based algorithm was run only once, since it is deterministic.
The initial value of $N$ for the SMT-based algorithm was set to the size of the smallest CLAs
among those obtained by the 10 runs of the proposed algorithm.
This favors the SMT-based algorithm, since it ensures that the output CLA of the SMT-based algorithm
is never greater in size than those obtained by the proposed heuristic algorithm.

All the experiments were conducted on a machine with 3.2GHz 8-Core Intel Xeon W CPU and 128GB memory, running MacOS Mojave.
We wrote a C++ program that implements the proposed algorithm.
The CIT-BACH tool~\cite{CIT-BACH} was used as a 3-way CCA generator.
The implementation of the SMT-based algorithm was done using C.
The Yices SMT solver~\cite{Dutertre:cav2014} was used in this implementation.
The timeout period was set to 1 hour for every run of both algorithms.

The results of the experiments are shown in Table \ref{tab:compare}.
The two rightmost columns of the table show the benchmark IDs and names.
The third and forth columns show the number of factors and the number of valid two-way interactions
for each benchmark.
The fifth column, marked with an asterisk (*), shows the number of unordered pairs $T_a, T_b (\neq T_a) \in \mathcal{VI}_2$ such that $\{T_a\}$ and $\{T_b\}$
are indistinguishable.

The remaining part of the table is divided into two parts:
one for the proposed algorithm and the other for the SMT-based generation algorithm.
In the proposed algorithm part, the left three columns show the maximum, minimum, and average sizes of the generated $(\overline{1} ,\overline{2})$-CLAs.
In the column labeled ``Average (3-CCA)'', the figures in parentheses indicate the sizes of the 3-way CCAs generated by \textsc{generateCCA()} on Line~\ref{alg:ccagen} in Algorithm~\ref{alg:main}.
The next three columns indicate the maximum, minimum, and average running times.
The running time is the sum of the time used for generating 3-way CCAs and the time used for deleting redundant rows from those 3-way CCAs.
The unit is seconds.

The two rightmost columns show the results of the SMT-based algorithm.
They show, for each problem instance, the size of the smallest CLA obtained and
the running time taken by the algorithm to produce that CLA.
(Thus, the running time does not include the running time of runs with $N>N_{sm}$ and $N < N_{sm}$,
where $N$ is the given size of an array and $N_{sm}$ denotes the size of the smallest CLA.)
As stated above, the algorithm was iterated with decreasing $N$ until 
it failed to solve the problem within the timeout period 
or proves the nonexistence of a CLA of size $N$.
In the latter case, the CLA obtained in the immediately previous iteration
is guaranteed to be optimal in size.
The figures in bold font show the sizes of these optimal $(\overline{1} ,\overline{2})$-CLAs.
The ``T.O.'' marks indicate that even the first iteration with the initial $N$ was not completed
because of timeout.

\begin{sidewaystable*}
  \centering
  \caption{Experimental results that compare  
  	CLA sizes and running times between the proposed heuristic algorithm and the SMT-based algorithm}
  \label{tab:compare}
  \scalebox{0.8}{
    \begin{tabular}{c|c|c|c|c|c|c|c|r|r|r|c|r}
\hline 
\multirow{3}{*}{No.} & \multirow{3}{*}{SUT} & \multirow{3}{*}{$|\mathcal{F}|$} & \multirow{3}{*}{$|\mathcal{VI}_2|$} & \multirow{3}{*}{{*}} & \multicolumn{6}{c|}{proposed method} & \multicolumn{2}{c}{SMT based method}\tabularnewline
\cline{6-13} \cline{7-13} \cline{8-13} \cline{9-13} \cline{10-13} \cline{11-13} \cline{12-13} \cline{13-13} 
 &  &  &  &  & \multicolumn{3}{c|}{$(\overline{1},\overline{2})$-CLA Sizes} & \multicolumn{3}{c|}{time(s)} & \multirow{2}{*}{$(\overline{1},\overline{2})$-CLA} & \multirow{2}{*}{time(s)}\tabularnewline
\cline{6-11} \cline{7-11} \cline{8-11} \cline{9-11} \cline{10-11} \cline{11-11} 
 &  &  &  &  & Max & Min & Average(3-CCA) & Max & Min & Average &  & \tabularnewline
\hline 
1 & Aircraft & 11 & 180 & 54 & 17 & 15 & 16.3 (23.6) & 0.28 & 0.13 & 0.15 & 13 & 25.60\tabularnewline
2 & Car & 9 & 102 & 161 & {\bf 10} & {\bf 10} & {\bf 10.0} (12.2) & 0.13 & 0.12 & 0.12 & {\bf 10} & 0.07\tabularnewline
3 & Movie & 13 & 178 & 567 & 9 & {\bf 8} & 8.4 (11.2) & 0.14 & 0.14 & 0.14 & {\bf 8} & 0.07\tabularnewline
4 & Medicitries & 3 & 41 & 0 & 25 & 23 & 24.3 (41.0) & 0.11 & 0.10 & 0.10 & 23 & 0.98\tabularnewline
5 & Medicitries\_small & 3 & 58 & 1 & 35 & 32 & 33.4 (70.0) & 0.12 & 0.10 & 0.11 & 30 & 47.70\tabularnewline
6 & Banking1 & 5 & 102 & 0 & 28 & 25 & 26.7 (61.3) & 0.14 & 0.13 & 0.14 & 23 & 913.41\tabularnewline
7 & Banking2 & 15 & 473 & 0 & 30 & 27 & 27.5 (41.2) & 0.17 & 0.17 & 0.17 & T.O. & T.O.\tabularnewline
8 & Concurrency & 5 & 36 & 16 & {\bf 7} & {\bf 7} & {\bf 7.0} (8.0) & 0.11 & 0.10 & 0.10 & {\bf 7} & 0.01\tabularnewline
9 & CommProtocol & 11 & 285 & 69 & 35 & 33 & 34.0 (54.6) & 0.19 & 0.18 & 0.18 & T.O. & T.O.\tabularnewline
10 & Healthcare1 & 10 & 361 & 5 & 50 & 45 & 48.2 (125.1) & 0.21 & 0.19 & 0.20 & 42 & 3,438.74\tabularnewline
11 & Healthcare2 & 12 & 466 & 0 & 36 & 33 & 34.9 (76.8) & 0.21 & 0.20 & 0.20 & T.O. & T.O.\tabularnewline
12 & Healthcare3 & 29 & 3,092 & 477 & 101 & 77 & 91.9 (251.4) & 4.59 & 4.03 & 4.38 & T.O. & T.O.\tabularnewline
13 & Healthcare4 & 35 & 5,707 & 288 & 105 & 98 & 101.7 (379.3) & 21.59 & 20.34 & 20.85 & T.O. & T.O.\tabularnewline
14 & Insurance & 14 & 4,573 & 0 & 805 & 789 & 794.2 (7,325.3) & 84.94 & 83.63 & 84.34 & T.O. & T.O.\tabularnewline
15 & NetworkMgmt & 9 & 1,228 & 0 & 210 & 202 & 207.2 (1,199.8) & 1.89 & 1.84 & 1.86 & T.O. & T.O.\tabularnewline
16 & ProcessorComm1 & 15 & 1,058 & 6 & 63 & 58 & 59.9 (164.6) & 0.62 & 0.56 & 0.59 & T.O. & T.O.\tabularnewline
17 & ProcessorComm2 & 25 & 2,525 & 1,562 & 68 & 65 & 66.8 (200.3) & 2.14 & 1.99 & 2.04 & T.O. & T.O.\tabularnewline
18 & Services & 13 & 1,819 & 93 & 200 & 194 & 197.0 (1,258.6) & 4.70 & 4.52 & 4.62 & T.O. & T.O.\tabularnewline
19 & Storage1 & 4 & 53 & 11 & 22 & 22 & 22.0 (25.0) & 0.12 & 0.11 & 0.11 & 22 & 70.63\tabularnewline
20 & Storage2 & 5 & 126 & 0 & 37 & 34 & 35.9 (78.2) & 0.12 & 0.11 & 0.12 & 30 & 564.92\tabularnewline
21 & Storage3 & 15 & 1,020 & 57 & 89 & 87 & 87.5 (269.4) & 0.81 & 0.75 & 0.77 & T.O. & T.O.\tabularnewline
22 & Storage4 & 20 & 3,491 & 0 & 222 & 215 & 218.6 (1,183.1) & 15.57 & 14.92 & 15.30 & T.O. & T.O.\tabularnewline
23 & Storage5 & 23 & 5,342 & 20 & 361 & 344 & 355.7 (2,137.5) & 56.71 & 55.82 & 56.18 & T.O. & T.O.\tabularnewline
24 & SystemMgmt & 10 & 310 & 130 & 31 & 27 & 29.0 (66.8) & 0.17 & 0.15 & 0.16 & T.O. & T.O.\tabularnewline
25 & Telecom & 10 & 440 & 23 & 54 & 49 & 51.5 (144.5) & 0.24 & 0.22 & 0.22 & T.O. & T.O.\tabularnewline
26 & Apache & 172 & 66,927 & 0 & 89 & 85 & 86.6 (232.4) & 3,350.32 & 3,138.48 & 3,263.15 & T.O. & T.O.\tabularnewline
27 & Bugzilla & 52 & 5,818 & 0 & 48 & 41 & 45.3 (80.8) & 9.80 & 7.98 & 9.02 & T.O. & T.O.\tabularnewline
28 & GCC & 199 & 82,770 & 46 & 68 & 61 & 64.4 (128.8) & 2,618.84 & 2,286.57 & 2,494.59 & T.O. & T.O.\tabularnewline
29 & Spins & 18 & 979 & 9 & 53 & 49 & 50.7 (136.3) & 0.64 & 0.61 & 0.63 & T.O. & T.O.\tabularnewline
30 & Spinv & 55 & 8,741 & 599 & 97 & 91 & 93.2 (321.4) & 59.99 & 56.37 & 58.21 & T.O. & T.O.\tabularnewline
\hline 
\end{tabular}

  }
\end{sidewaystable*}

\subsection{Experimental Results}

\paragraph{Answer to RQ~1}
The proposed heuristic algorithm was able to find CLAs for all the benchmarks.
The running time was even less than one second for many of these.
Except for the two largest problem instances, it was at most 90 seconds.
The two exceptional instances are Apache and GCC, both having nearly 200 factors.
Even for these large benchmarks, the algorithm terminated, successfully producing CLAs
within the one hour time limit.
The proposed algorithm was able to generate CLAs that are considerably smaller
than the initial CCAs.
The reduction rate varies for different problem instances;
but it was greater than 50\% for many of the problems.
Even a more than five-fold reduction was observed for some benchmarks,
namely, Insurance (No.~14), NetworkMgmt (No.~15), Services (No.~18),
Storage4 (No.~22), and Storage5 (No.~23).
In summary, the proposed heuristic algorithm is able to generate CLAs
within a reasonable time unless the problem is not very large.
The sizes of CLAs produced by the algorithm are substantially smaller than the initial
3-CCAs.

\paragraph{Answer to RQ~2}
When comparing the running times between both algorithms, the proposed algorithm
shows distinguishing results. For all benchmarks except Car, Movie, Concurrency,
the proposed algorithm achieved orders of magnitude reduction.
The SMT-based algorithm often timed out
even for the benchmarks that the proposed algorithm solved in less than one second.
The difference can be explained as follows.
To generate a CLA, the SMT-based algorithm needs to solve a constraint satisfaction
problem represented by logic expressions.
This problem can be very difficult to solve, especially when the given number of rows, $N$,
approaches to the lower limit of the size of CLAs.
On the other hand, the proposed heuristic simply repeats the check-and-delete process
until all rows are examined.
In the experiments, as stated above, we set the initial $N$ of the SMT-based algorithm 
to the size of the smallest CLA obtained by 10 runs of the proposed heuristic
algorithm. Hence the sizes of the CLAs generated by the SMT-based algorithm were guaranteed not to exceed 
those generated by the proposed heuristic algorithm.
The experimental results show that the SMT-based algorithm was often successful in
further decreasing the sizes of CLAs by, typically, a few rows.
This also suggests that the proposed algorithm rarely produces the minimum (optimal) CLAs.
One possible reason for this is that 3-way CCAs generated by \textsc{generateCCA()}
may not be a superset of any of the optimal CLAs.
Another reason is that resulting CLAs depends on the order of deleting rows.
As there are a number of deleting orders, it can be unlikely
that the one that leads to the optimal CLA, if any, is selected.
In summary, the proposed heuristic algorithm runs much faster than
does the SMT-based algorithm.
If the problem is small enough for the SMT-based algorithm to
handle, the algorithm is superior in yielding small CLAs to the proposed heuristic
algorithm.

\paragraph{Answer to RQ~3}
As stated, the proposed algorithm was able to produce CLAs
in very short time for many problem instances.
Even for very large benchmarks, namely, Apache and GCC,
it completed generation of CLAs within one hour.
These benchmarks are model taken from the real-world applications.
Hence we conclude that, although further improvement is still desirable,
the proposed algorithm can scale to real-world problems.

\section{Related Work}\label{sec:related}

Constraint handling has been an important issue in combinatorial 
interaction testing, even before the name of this testing approach was 
coined. 
Early work includes, for example,~\cite{Tatsumi87,AETG1997}. 
There is even a systematic literature review that is dedicated to constraint handling 
in combinatorial interaction testing~\cite{8102999}.
This literature review lists 103 research papers addressing this particular problem. 

In contrast, research on LAs is still in an early stage~\cite{7528941}. 
The notion of LAs was originally proposed by Colbourn and McClary~\cite{colbourn_locatingarray2008}. 
Since then, some studies have been published that discuss mathematical properties of LAs 
or propose mathematical constructions of LAs. 
These studies include~\cite{Shi2012,tang_optimalitylocatingarray_2012,Spread2016,colbourn2016}. 
Some other studies propose computational generation methods of LAs~\cite{nagamoto_locatingpairwisetesting2014,Konishi2017,DBLP:journals/corr/abs-1904-07480,10.1007/978-3-319-94667-2_29}. 
None of these previous studies consider constraints. 
Recent surveys on the state of locating array research and its applications can be found in 
\cite{Colbourn2018,7528941}.

Mathematical objects similar to LAs include 
\textit{Detecting Arrays}~\cite{colbourn_locatingarray2008,ShiDetecting2012} 
and \textit{Error Locating Arrays}~\cite{ELA2010}.   
To our knowledge, no attempts have been reported to incorporate constraints 
into these arrays, either. 

We for the first time introduced the concept of CLA in~\cite{DBLP:journals/corr/abs-1801-06041}, which is a preprint of an early version 
of this paper. 
This paper extends the early version by incorporating our subsequent work~\cite{8639696}, 
where we showed the heuristic algorithm for obtaining CLAs for the first time. 
Originally we presented it as a method of generating $(\overline{1}, t)$-CLAs, 
instead of $(\overline{1}, \overline{t})$-CLAs. 
This paper extends \cite{8639696}  
by providing new theorems (namely, Theorems~\ref{theorem:cca2cla} and~\ref{theorem:t2bart}) to show that the algorithm can yield $(\overline{1}, \overline{t})$-CLAs 
and 
by providing more comprehensive experimental results using a new, faster implementation of the algorithm.   
The SMT-based algorithm, which was compared with the proposed algorithm in Section~\ref{sec:evaluation}, 
was presented in~\cite{Jin2018}. 

There are studies that address fault location and analysis of test execution results 
(often referred to as fault characterization) without using the mathematical objects mentioned above. 
The studies in this line of research include, for example, \cite{Niu2018TSE,Shakya2012,Zhang2011,Li2012,Wang2010,Fouche2009,Yilmaz2006,Ghandehari2013,Nishiura2017}. 

\section{Conclusions}\label{sec:summary}

In this paper, we introduced the notion of Constrained Locating Arrays (CLA), 
which generalize locating arrays by incorporating constraints on test parameters 
into them. The extension enables to apply locating arrays to testing of 
real-world systems which usually have such constraints.  
We proved some basic properties of CLAs and then 
presented a heuristic algorithm to generate CLAs that can locate at most one
failure triggering interaction. 
Experimental results using a number of practical problem instances 
showed that the proposed algorithm is able to construct CLAs with reasonable time. 
Possible future research directions include, for example, 
developing other algorithms for CLA generation and 
investigating the usefulness of CLAs in real-world testing. 

\appendix

\section{Proof of Theorem~\ref{theorem:t2bart}}

\begin{lemma}\label{lemma:cla2cca}
	Suppose that an $N \times k$ array $A$ is a $(\overline{1}, t)$-CLA such that $1\leq t \leq k$.  
	Then $A$ is a $t$-CCA.
\end{lemma}
\begin{proof}
	Since $A$ is a $(\overline{1}, t)$-CLA, $\rho_A(\mathcal{T}_1) \neq \rho_A(\mathcal{T}_2)$ 
	for any $\mathcal{T}_1, \mathcal{T}_2 (\neq \mathcal{T}_1) \subseteq \mathcal{VI}_t$ such that 
	$|\mathcal{T}_1|, |\mathcal{T}_2| \leq 1$.
	Hence, if $\mathcal{T}_1 = \emptyset$ and $\mathcal{T}_2 = \{T\}$ for any $T \in \mathcal{VI}_t$, 
	then $\rho_A(\mathcal{T}_1) = \rho_A(\emptyset) = \emptyset \neq \rho_A(\mathcal{T}_2) = \rho_A(T)$. 
\end{proof}

\noindent
\textbf{Theorem~\ref{theorem:t2bart}} (in Section~\ref{sec:generation})\textbf{.}
{\it 
	If an $N \times k$ array $A$ is a $(\overline{1}, t)$-CLA such that $1 \leq t \leq k$, then 
$A$ is a $(\overline{1}, \overline{t})$-CLA. 
}

\begin{proof}
	Suppose that $A$ is a $(\overline{1}, t)$-CLA such that $1 \leq t \leq k$. 
	By Lemma~\ref{lemma:cla2cca},  $A$ is a $t$-CCA; thus, by Theorem~\ref{theorem:cca2cla}, 
	it is a $(\overline{1}, \overline{t-1})$-CLA. 
	Recall that $A$ is a $(\overline{1}, \overline{t})$-CLA iff 
	$\rho_A(\mathcal{T}_1) \neq \rho_A(\mathcal{T}_2)$ for all 
	$\mathcal{T}_1, \mathcal{T}_2 \in \overline{\mathcal{VI}_t}$ such that
	$\mathcal{T}_1$ and $\mathcal{T}_2$ are distinguishable
	and  
	$0\leq |\mathcal{T}_1|, |\mathcal{T}_2| \leq 1$. 
	(Note that $\mathcal{T}_1$ and $\mathcal{T}_2$ are trivially independent.)
	If $ |\mathcal{T}_1| = |\mathcal{T}_2| = 0$, then $\mathcal{T}_1$ and $\mathcal{T}_2$ are 
	both $\emptyset$ and thus indistinguishable.
	If $ |\mathcal{T}_1| = 0$ and $|\mathcal{T}_2| = 1$, then $\mathcal{T}_2 = \{T\}$ 
	for some $T \in \overline{\mathcal{VI}_t}$. 
	Since $A$ is a $t$-CCA, $\rho_A(\{T\}) \neq \emptyset$ for any $T \in \overline{\mathcal{VI}_t}$. 
	Therefore $\rho_A(\mathcal{T}_1) \neq \rho_A(\mathcal{T}_2)$. 
	Clearly this argument holds when $ |\mathcal{T}_1| = 1$ and $|\mathcal{T}_2| = 0$.
	
	In the following part of the proof, we assume that $ |\mathcal{T}_1| = |\mathcal{T}_2| = 1$. 
	Let $\mathcal{T}_1 = \{T_a\}$, $\mathcal{T}_2 = \{T_b\}$ where $T_a, T_b \in \overline{\mathcal{VI}_t}$.
	Without losing generality, we assume that the strength of $T_a$ is at most equal to that of $T_b$, 
	i.e., $0 \leq |T_a|\leq |T_b| \leq t$.
	If $0 \leq |T_a|\leq |T_b| \leq t-1$ and $\{T_a\}$ and $\{T_b\}$ are distinguishable, 
	then $\rho_A(\mathcal{T}_1) \neq \rho_A(\mathcal{T}_2)$ since 
	$A$ is a $(\overline{1}, \overline{t-1})$-CLA. 
	If $|T_a| = |T_b| = t$ and $\{T_a\}$ and $\{T_b\}$ are distinguishable, 
	then $\rho_A(\mathcal{T}_1) \neq \rho_A(\mathcal{T}_2)$ since 
	$A$ is a $(\overline{1}, t)$-CLA. 
	
	Now consider the remaining case where $0 \leq |T_a| < |T_b| = t$. 
%
%
	Assume that $\{T_a\}$ and $\{T_b\}$ are distinguishable. 
    Below we show that $\rho_A(\mathcal{T}_1) \neq \rho_A(\mathcal{T}_2)$ 
    under this assumption. 
	Because of the assumption, at least either one of the following two cases holds:   
	Case 1: for some $\boldsymbol{\sigma} \in \mathcal{R}$, $T_a \subseteq \sigma$ and $T_b \not\subseteq \boldsymbol{\sigma}$, or
	Case 2:  
	for some $\boldsymbol{\sigma} \in \mathcal{R}$, 
	$T_a \not\subseteq \boldsymbol{\sigma}$ and $T_b \subseteq \boldsymbol{\sigma}$.
	
	Let $T_a=\{(F_{a_1},u_{a_1}),\dots,(F_{a_l},u_{a_l})\}$ and $T_b=\{(F_{b_1},v_{b_1}),\dots,(F_{b_t},v_{b_t})\}$
	($0\leq l \leq t-1$).
	Also let ${\rm F}=\{F_{a_1},\dots,F_{a_l}\}\cap\{F_{b_1},\dots,F_{b_t}\}$; i.e.,
	${\rm F}$ is the set of factors that are involved in both interactions.
	
	Case~1:  Let $\boldsymbol{\sigma}_1$ be any test in $\mathcal{R}$ such that $T_a \subseteq \boldsymbol{\sigma}_1$ and $T_b \not\subseteq \boldsymbol{\sigma}_1$. 
	Choose a factor $F_{b_i}$, $1 \leq i \leq t$ such that 
	the value on $F_{b_i}$ in $\boldsymbol{\sigma}_1$ is different from $v_{b_i}$. 
    Such a factor must always exist, because otherwise $T_b \subseteq {\boldsymbol{\sigma}_1}$.  
	Let $w_{b_i}$ denote the value on $F_{b_i}$ in $\boldsymbol{\sigma}_1$.
	Then interaction $\hat T = T_a \cup \{(F_{b_i}, w_{b_i})\}$ is 
	covered by $\boldsymbol{\sigma}_1$ ($\hat T \subseteq \boldsymbol{\sigma}_1$) and thus is valid. 
	The strength of $\hat T$ is  
	$l$ (if $F_{b_i} \in \mathrm{F}$,  in which case $u_{b_i} = w_{b_i}$) or $l+1$ (if $F_{b_i} \not\in \mathrm{F}$). 
	For any test $\boldsymbol{\sigma} \in \mathcal{R}$,  $\hat T \subseteq \boldsymbol{\sigma} \Rightarrow T_b \not\subseteq \boldsymbol{\sigma}$ holds
	because $w_{b_i}  \not= v_{b_i}$. 
	Since $A$ is a $t$-CCA and the strength of $\hat T$ is at most $t$, $A$ has a row that covers $\hat T$. 
	This row covers $T_a$ but not $T_b$; thus $\rho_A(\mathcal{T}_1) \neq \rho_A(\mathcal{T}_2)$.
	
	Case 2:
	Let $\boldsymbol{\sigma}_2$ be any test in $\mathcal{R}$ such that $T_a \not \subseteq \boldsymbol{\sigma}_2$ and $T_b \subseteq \boldsymbol{\sigma}_2$.
    Also let $\check T$ be any $t$-way interaction such that $\check T = T_a \cup \{(F_{b_{i_1}}, v_{b_{i_1}})$, $\ldots$, $(F_{b_{i_{t-l}}}, v_{b_{i_{t-l}}})\}$ for some $F_{b_{i_{1}}}, \dots, F_{b_{i_{t-l}}} \not \in F$. 
    In other words, $\check T$ is a $t$-way interaction that is obtained by extending $T_a$ with some $t-l$
    factor-value pairs in $T_b$. 

If $\check T$ is valid, then $\{\check T\}$ and $\{T_b\}$ are distinguishable, 
because $T_b \subseteq \boldsymbol{\sigma}_2$ and $\check T \not \subseteq \boldsymbol{\sigma}_2$ (since $T_a \not \subseteq \boldsymbol{\sigma}_2$ and $T_a \subseteq \check T$). 
$A$ is a $(\overline{1}, t)$-CLA; thus $A$ must have a row $\boldsymbol{r}$ that covers either $\check T$ or $T_b$; 
i.e., $\check T \subseteq r \land T_b \not\subseteq \boldsymbol{r}$ or $\check T \not \subseteq \boldsymbol{r} \land T_b \subseteq \boldsymbol{r}$.
$\check T \subseteq \boldsymbol{r} \land T_b \not\subseteq \boldsymbol{r}$ directly implies $T_a \subseteq \boldsymbol{r} \land T_b \not\subseteq \boldsymbol{r}$,
while $\check T \not \subseteq \boldsymbol{r} \land T_b \subseteq \boldsymbol{r}$ implies $\check T \backslash T_b \not\subseteq \boldsymbol{r}$, which means $T_a \not \subseteq \boldsymbol{r}$. 
Hence $\rho_A(\mathcal{T}_1) \neq \rho_A(\mathcal{T}_2)$.

If $\check T$ is not valid, then we can show that $T_a$ and $T_b$ never appear simultaneously in any test $\boldsymbol{\sigma} \in \mathcal{R}$ as follows. 
If there is some test $\boldsymbol{\sigma}$ in $\mathcal{R}$  in which $T_a$ and $T_b$ are both covered, then $\check T$ is also covered by some tests (including $\boldsymbol{\sigma}$ ) in $\mathcal{R}$; i.e., $\check T$ is valid.  
The contraposition of this argument is that if $\check T$ is invalid, then there is no test in $\mathcal{R}$ that covers $T_a$ and $T_b$. 
Since $A$ is a $t$-CCA and $T_a, T_b \in \overline{\mathcal{VI}_t}$, $\rho_A(T_a) \neq \emptyset$ and $\rho_A(T_b) \neq \emptyset$.  
Hence $\rho_A(\mathcal{T}_1) \neq \rho_A(\mathcal{T}_2)$.
\end{proof}





%
%
%
\end{document}